\documentclass[letterpaper, 10 pt, conference]{ieeeconf}
\usepackage{cite}
\usepackage{amsmath,amssymb,amsfonts}
\usepackage{graphicx}
\usepackage{textcomp}
\usepackage{threeparttable}
\usepackage{hyperref}
\usepackage{amssymb}
\usepackage{enumitem}
\usepackage{tabu}
\usepackage[utf8]{inputenc}
\usepackage{setspace}
\usepackage{subcaption}
\usepackage{multirow}
\usepackage{tabularx}
\usepackage{array}
\usepackage{makecell}
\usepackage{algpseudocode}
\usepackage{dblfloatfix}
\usepackage{tikz-network}
\tikzset{
block/.style = {draw, fill=white, rectangle, minimum height=2em, minimum width=5em},
tmp/.style  = {coordinate}, 
sum/.style= {draw, fill=white, circle, node distance=1cm},
input/.style = {coordinate},
output/.style= {coordinate},
pinstyle/.style = {pin edge={to-,thin,black}
}
}
\newtheorem{theorem}{\textbf{Theorem}}
\newtheorem{lemma}{\textbf{Lemma}}

\newtheorem{corollary}{\textbf{Corollary}}
\newtheorem{remark}{\textbf{Remark}}
\usepackage{authblk}
\usepackage{float}
\usepackage[table]{xcolor}
\newtheorem{definition}{\textbf{Definition}}

\usepackage{color}
\usepackage{tikz}  
\IEEEoverridecommandlockouts 
\newcommand{\absElem}[1]{|#1 |_{\text{abs}}}

\title{\Large Robust Stability Analysis of Positive Lur'e System with Neural Network Feedback}

\author{Hamidreza Montazeri Hedesh$^1$, Moh. Kamalul Wafi$^1$, Bahram Shafai$^1$, and Milad Siami$^1$\thanks{$^1$H. Montazeri Hedesh, M. K. Wafi, B. Shafai and M. Siami are with the Department of Electrical \& Computer Engineering, Northeastern University, Boston, MA 02115, USA.
(e-mails: {\tt\footnotesize	\{montazerihedesh.h, wafi.m, b.shafai m.siami\}@northeastern.edu}).
}

\thanks{This material is based upon work supported in part by the U.S. Office of Naval Research under Grant Award N00014-21-1-2431 and in part by the DEVCOM Analysis Center and was accomplished under Contract Number W911QX-23-D0002. The views and conclusions contained in this document are those of the authors and should not be interpreted as representing the official policies, either expressed or implied, of the DEVCOM Analysis Center or the U.S. Government. The U.S. Government is authorized to reproduce and distribute reprints for Government purposes notwithstanding any copyright notation herein.
}
}
\begin{document}
\maketitle
\thispagestyle{empty}
\pagestyle{empty}
\begin{abstract}
This paper investigates the robustness of the Lur'e problem under positivity constraints, drawing on results from the positive Aizerman conjecture and robustness properties of Metzler matrices. Specifically, we consider a control system of Lur'e type in which not only the linear part includes parametric uncertainty but also the nonlinear sector bound is unknown. We investigate tools from positive linear systems to effectively solve the problems in complicated and uncertain nonlinear systems.
By leveraging the positivity characteristic of the system, we derive an explicit formula for the stability radius of Lur'e systems. Furthermore, we extend our analysis to systems with neural network (NN) feedback loops. Building on this approach, we also propose a refinement method for sector bounds of NNs. This study introduces a scalable and efficient approach for robustness analysis of both Lur'e and NN-controlled systems. Finally, the proposed results are supported by illustrative examples.
\end{abstract}

\allowdisplaybreaks
\section{Introduction}
In recent years, the rapid advancement of NNs has driven their integration into various safety critical autonomous systems, raising critical concerns about their stability and robustness in real-world applications \cite{szegedy2013intriguing}. Therefore, ensuring their reliable performance under uncertainty is imperative. Researchers strived to address this issue by adopting various concepts from mathematics and control theory. Surveys such as \cite{dawson2023safe} and \cite{tsukamoto2021contraction} explain some of the efforts. However, a very straightforward yet resourceful tool has been overlooked in these efforts—the Lur'e system.

Lur’e systems provide a flexible model for representing feedback systems with both system uncertainties and complicated nonlinearities like NNs, enabling the use of various stability analysis tools. This paper leverages the Lur’e framework to analyze the robustness of NN-controlled systems. We first explore the robustness of Lur'e systems under positivity constraints and structured uncertainties, then apply the framework to NN-controlled systems to enhance safety in applications.

Lur'e systems have long served as a foundation for stability analysis of nonlinear systems, underpinning fundamental methods such as the \textit{Lur'e-Postnikov criterion}, \textit{Aizerman and Kalman conjectures}, the \textit{circle criterion}, and their extensions. However, these classical approaches have limitations when dealing with system uncertainties, leaving room for further exploration.
The existing literature on robustness analysis of Lur'e systems with structured uncertainties, which is the focus of our study, spans various approaches. For instance, studies such as \cite{TESI1991147, wang2008reliable, hao2010absolute} employ Popov criterion, $H_\infty$ control, and Lyapunov theory, respectively. Additionally, some works, like \cite{tan1999absolute}, explore more complex perturbation structures. There are recent works such as \cite{yin2021stability} that implicitly studied robust stability of Lur'e systems using Integral Quadratic Constraints (IQCs) and implemented it on an NN-controlled system. However, their proposed method suffers from high computational complexity when it comes to large systems. In contrast, this paper introduces a straightforward and scalable approach to robustness analysis of Lur'e systems. By utilizing positivity constraints on the system, we derive an explicit formula for the stability radius of Lur’e systems under a comprehensive class of structured uncertainties, improving upon prior works with similar uncertainty structure such as \cite{fuh2009robust, dun2008robust} in both simplicity and computational efficiency.

It is worth noting that positive-constrained systems are widely used in modeling numerous real-world systems, including engineering, biology, economics, and population dynamics.
The positivity constraint offers a significant advantage in the stability and robustness analysis of linear systems \cite{farina2000positive,10708136,shafai2024positive}; however, its integration into nonlinear systems has not been extensively explored. In this paper, we exploit positivity constraints to develop efficient robustness analysis tools for nonlinear Lur’e systems under structured uncertainty. A key novelty of this paper is demonstrating how a linear system tool can be extended to analyze nonlinear robust stability.

Lur'e systems have also been employed in the safety analysis of NNs. For instance, studies such as \cite{fazlyab2021introduction,yin2021stability, yin2021imitation} implicitly utilize Lur'e systems alongside the \textit{S-procedure} to define the forward reachable set of an NN and ensure the stability of NN-controlled systems. Our recent studies \cite{10561541,hedesh2025local} introduced a simple, scalable method for stability assurance in NN-controlled systems using the Lur'e framework. The studies \cite{Wafi,mashhadireza2024iterative} implicitly utilize Lur'e system with reinforcement learning and iterative learning subject to time delays and nonlinear perturbations, respectively. Also, some studies \cite{10155284} investigated the stability of feedback systems with ReLU NNs. However, none of them addressed a scalable and formally guaranteed robustness measure for NN feedback loops.  To address this gap, our study proposes a straightforward and scalable method for calculating the stability radius of systems with NN feedback loops.

In addition to robustness analysis, this study addresses sector bounds for NNs. Sector bounds are a classical tool in control theory, used to ensure stability in systems with nonlinearities by defining constraints within which the system's nonlinearity must operate. NN bounds are also one of the leading topics in the machine learning community. There are extensive recent studies on NN bounding methods and refinement of the existing bounds. Examples are the evolution of bounding methods like \textit{CROWN} \cite{kotha2023provably}, \textit{BnB} \cite{zhang22babattack}, \textit{DeepPoly} \cite{singh2019abstract}, \textit{Recurjac} \cite{zhang2019recurjac}, and \textit{global sector bounds} \cite{10561541}. In this study, we build upon our results in Lur'e systems and propose a novel method for obtaining tight sector bounds for NNs.

The contributions of this paper are threefold. First, we propose an explicit formula for the stability radius of Lur'e systems with positivity constraints under a class of structured uncertainties. Second, by defining sector bounds for FFNNs, we extend this analysis to systems with NN feedback loops. Finally, we apply these findings to establish refined sector bounds for NNs in the feedback loop of control systems.

This paper is structured as follows. First, we introduce the theoretical frameworks that underpin our main results. The next two sections present our core findings, followed by illustrative examples demonstrating their applicability. Finally, we summarize the key insights and conclude with a discussion of our results.

\subsection{Notation}
The set of real numbers, \(n\)-dimensional real vectors, and \(m \times n\)-dimensional real matrices are denoted by \(\mathbb{R}\), \(\mathbb{R}^n\), and \(\mathbb{R}^{m \times n}\), respectively. Additionally, $\mathbb{C}$ is set of complex numbers.
The standard comparison operators \( < \), \( \geq \), \( \leq \), and \( > \) are applied elementwise to vectors and matrices.
\(\mathbb{R}_+\) denotes the set of nonnegative real numbers. Moreover, \(\mathbb{R}_+^n := \{\nu \in \mathbb{R}^n : \nu \geq 0\}\) and \(\mathbb{R}_+^{n \times m} := \{A \in \mathbb{R}^{n \times m} : A \geq 0\}\) represents the set of vectors and matrices with real nonnegative elements. A \emph{Metzler} matrix is characterized by having nonnegative off-diagonal entries. An \textit{M}-matrix is defined as a matrix with non-positive off-diagonal elements and eigenvalues with non-negative real parts. A vector of zeros of size \(n\) is denoted by \(\mathbf{0}_n\). For a given matrix \(A\), the order $\absElem{A}$ means the elementwise absolute value of matrix \(A\), and $\|A\|$ denotes a matrix norm of choice.

\section{Proposed Theoretical Framework}
\begin{figure}[t!]
    \centering
    \begin{tikzpicture}[auto, node distance=2cm,>=latex']
        \node [tmp] (tmp1) {};
        \node [tmp, below of=tmp1,node distance=.75cm] (tmp) {};
        \node [block, right of=tmp1,node distance=1.5cm] (Control) {$\Phi(y)$};
        \node [block, right of=Control,node distance=2.5cm] (Plant) {$G$};
        \node [tmp, right of=Plant,node distance=2.5cm] (tmp2) {};
        \node [tmp, right of=Plant,node distance=1.5cm] (tmp3) {};
        \draw [->] (Control) -- node{$u$} (Plant);
        \draw [->] (Plant) -- node{$y=Cx$} (tmp2);
        \draw [-] (tmp3) |- (tmp);
        \draw [->] (tmp) |- (Control);
    \end{tikzpicture}
    \caption{\small Lur'e system with plant $G$ and nonlinear controller $\Phi$.}\vspace{-.4cm}
    \label{fig:luresystem}
\end{figure}
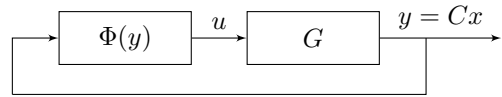
In this section, we establish the theoretical frameworks that underpin our study, including the positive Lur'e system, positive Aizerman conjecture, and the stability radii of nonnegative and Metzler matrices.

\subsection{Positive Lur'e Systems}
Consider a linear time invariant (LTI) system of the form
\begin{equation}\label{eq:generallti}
    \dot x(t) = A x(t) + B u(t), \quad
    y(t) = C x(t),
\end{equation}
where $A \in \mathbb R^{n\times n}$, $B \in \mathbb R^{n\times m}$, and $C \in \mathbb R^{p\times n}$ are constant matrices, and $x(t) \in \mathbb R^n$, $u(t) \in \mathbb R^m$, and $y(t) \in \mathbb R^p$ denote state, input, and output variables.

\begin{definition}\label{def:positive system}
A linear system, described by \eqref{eq:generallti}, is said to be a \emph{positive} system if, for all \( t > 0 \), the state \( x(t) \) remains non-negative, given non-negative initial conditions \( x(0) \geq 0 \) and a non-negative input \( u(t) \geq 0 \).
\end{definition}

\begin{lemma}
The system \eqref{eq:generallti} is positive if and only if the matrix \( A \) is a Metzler matrix, $B \in \mathbb R_+^{n\times m}$, and $C \in \mathbb R_+^{p\times n}$.
\end{lemma}

\begin{lemma}
    A positive system of the form \eqref{eq:generallti} is asymptotically stable if and only if there exists a positive vector $v$ such that $Av<0$.
\end{lemma}

Now, we officially define the positive Lur'e system. Consider a system formed by the interconnection of the linear system \eqref{eq:generallti} and a static nonlinear feedback, represented as $u = \Phi(Cx,.)$. Fig. \ref{fig:luresystem} depicts a schematic of this system.
The overall dynamics of the system can be described by the following equation:
\begin{equation}\label{eq:positiveluresystem}
    \dot{x} = Ax(t) + B \Phi(Cx, t).
\end{equation}
Here, the function \( \Phi(\cdot): \mathbb{R}^p \times \mathbb{R} \to \mathbb{R}^m \) represents a static nonlinearity. Additionally, the matrices \( C \in \mathbb{R}_+^{p \times n} \) and \( B \in \mathbb{R}_+^{n \times m} \) have nonnegative elements, while no specific structural constraints, such as the Metzler property, are imposed on the system matrix \( A \).
We assume that $\Phi(0, t) = 0$ for all $t \geq 0$. This ensures that $x_* = 0$ is an equilibrium point of \eqref{eq:positiveluresystem}.
Additionally, we assume the existence of a unique, locally absolutely continuous function $\chi: \mathbb{R}_+ \to \mathbb{R}^n$ that satisfies \eqref{eq:positiveluresystem} almost everywhere for any initial condition $x(0) \in \mathbb{R}^n$. Moreover, we assume that $\Phi(z, t)$ is locally Lipschitz in $z$ and measurable in $t$, along with certain mild boundedness assumptions. These conditions guarantee the existence and uniqueness of the solution, as established in foundational control theory works such as \cite{sontag2013mathematical}.

As part of the definition of Lur'e systems, it is necessary to establish sector bounds for the nonlinear component of the system. In the context of Multi Input Multi Output (MIMO) systems, the multivariable nature of $\Phi$ necessitates a generalization of the classical sector bound definition used in Single Input Single Output (SISO) Lur'e systems. For positive systems, the sector bound for the multivariable function $\Phi$ is best defined through componentwise inequalities. The function $\Phi$ is said to be sector bounded in $[\Sigma_1, \Sigma_2]$ if the following condition holds:
\begin{equation}\label{eq:mimo sector bound}
    \Sigma_1 z \leq \Phi(z, t) \leq \Sigma_2 z, \quad \forall z \in \mathbb{R}^p_+, \forall t \geq 0,
\end{equation}
where $\Sigma_1, \Sigma_2 \in \mathbb{R}^{m \times p}$, and $\Sigma_1 \leq \Sigma_2$. The above formulation fully defines the MIMO positive Lur'e system under consideration. Based on this formulation, we proceed to introduce the positive Aizerman conjecture in the following subsection.

\subsection{Positive Aizerman Conjecture}
The Aizerman conjecture is a classical method for absolute stability analysis of general Lur'e systems, originally formulated for SISO systems. While counterexamples exist in the general case, the conjecture holds under stricter conditions such as system positivity \cite{drummond2022aizerman}. Furthermore, it can be extended to MIMO systems. By reducing nonlinear stability analysis to that of linear approximations, Aizerman conjecture offers a practical foundation for scalable and efficient stability analysis, particularly for systems with complex nonlinearities, such as NN-controlled systems.

The following theorem and remark confirm the validity of Aizerman conjecture for MIMO positive systems.

\begin{theorem}[\textbf{Positive Aizerman}\cite{drummond2022aizerman}]\label{the:positiveaizerman}
    Consider the Lur'e system described by \eqref{eq:positiveluresystem} where $B \geq 0$, $C \geq 0$, and the matrices $\Sigma_1$ and $\Sigma_2$ satisfy $\Sigma_1 \leq \Sigma_2$ with appropriate dimensions. If the matrix $A + B\Sigma_1 C$ is Metzler and $A + B\Sigma_2 C$ is Hurwitz, then for any nonlinearity $\Phi$ within the sector $[\Sigma_1, \Sigma_2]$, as defined by \eqref{eq:mimo sector bound}, the Lur'e system \eqref{eq:positiveluresystem} is globally exponentially stable.
\end{theorem}
\vspace{.1cm}
\begin{remark}\label{rem:stabrad}
    The Aizerman conjecture can be interpreted as a stability radius analysis of a linear system \(\dot{x} = Ax\) subjected to an additive, structured perturbation given by
    \[
    \dot{x} = Ax + B \Lambda(Cx)
    \]
    where \(\Lambda[.]\) is a placeholder for various classes of perturbations, from simple linear matrix multiplication to a nonlinear function \(\Lambda[Cz] = \Phi(Cz)\).

    The positive Aizerman conjecture asserts that if a system remains stable under all perturbations of a linear form \(\{B\Sigma C: \quad \forall \Sigma\in [\Sigma_1, \Sigma_2]\}\), then the system will also be stable for all corresponding nonlinear perturbations of the form \(\{B\Phi(Cx): \forall \Phi \in [\Sigma_1, \Sigma_2]\}\). In terms of stability radius, the positive Aizerman conjecture suggests that the "real static nonlinear stability radius," denoted \(r_{\mathbb{R}, \Phi}\), is equal to the "real linear stability radius," \(r_\mathbb{R}\). This implies that the stability radius of a Lur'e system can be calculated by the stability radius of its linearizations.
\end{remark}
\vspace{.1cm}
The proof of these results can be found in \cite{drummond2022aizerman}. The theorem and its remark establish a connection between nonlinear and linear systems. Our goal is to extend stability radius results from linear systems to nonlinear Lur'e systems.

Next, we present a fundamental result regarding the stability radius of Metzler matrices, which will be instrumental in the later stages of our analysis.

\subsection{Stability Radius of Positive Systems}
The stability radius refers to the largest amount of perturbation that a system can tolerate before it loses stability. Mathematically, it can be defined as the minimal distance (in terms of some norm) between the nominal system's matrix and the set of matrices that destabilize the system. For a system described by \(\dot x(t) = A x(t)\), the real stability radius is defined by
\begin{equation}\label{eq:rad}
    r_\mathbb{R} = \inf \left\{ \|\Delta\| : A + \Delta \text{ is unstable} \right\},
    \end{equation}
where \( A \in \mathbb{R}^{n\times n} \) is the system's nominal matrix. \( \Delta \in \mathbb{R}^{n\times n} \) represents a perturbation matrix and \( \|\Delta\| \) is a norm representing the size of the perturbation. 
Suppose the following LTI system with structured perturbation:
\begin{equation}\label{eq:disturbedsys}
    \dot{x}(t) = (A + D \Delta E)x(t),
\end{equation}
where the matrices \(A \in \mathbb{R}^{n \times n}\), \(D \in \mathbb{R}_+^{n \times m}\), and \(E \in \mathbb{R}_+^{p \times n}\) are given, and \(\Delta \in \mathbb{R}^{m \times p}\) is an unknown perturbation matrix confined to a certain set of interest.

We focus on the case where \(A\) is Metzler and stable. The stability assumption is included solely to guarantee the existence of the stability radius. Also, \(D \geq 0\) and \(E \geq 0\) are typical assumptions for systems with Metzler matrices and simplify the stability analysis.

The following theorem provides an explicit formula for the stability radius of Metzler matrices. For this class of matrices, it has been recognized that the real and complex stability radii are coinciding and given by a direct formula.

\begin{theorem}[\textbf{Metzler Stability Radius}\cite{shafai1997explicit}]\label{the:Metstabrad}
Consider the linear system \eqref{eq:disturbedsys}. Suppose that \(A\) is Metzler and stable, \(D \geq 0\) and \(E \geq 0\). Let $\|\cdot\|$ be a chosen operator norm on $\mathbb{F}^{m \times p},\mathbb{F}=\mathbb{R},\mathbb{C}$. Then, the  stability radii for system \eqref{eq:disturbedsys} under real perturbation $\Delta \in \mathbb{R}^{m\times p}$ $(r_\mathbb{R})$ or complex perturbations $\Delta \in \mathbb{C}^{m\times p}$ $(r_\mathbb{C})$ are given by
\begin{equation}\label{eq:robustradii}
\small
   r_{\mathbb{C}}(A,D,E) = r_{\mathbb{R}}(A,D,E) = \frac{1}{\| E (A)^{-1} D \|}.
\end{equation}
\end{theorem}
Furthermore, if $\Delta$ is defined by the set
$\Delta =\{S \odot \Delta : S_{ij} \geq 0\}$ with $\|\Delta\| = \max\{|\delta_{ij}|:\delta_{ij} \neq 0\}$ where $[S\odot\Delta]_{ij} = S_{ij}\delta_{ij}$ represents the Schur product, then
\begin{equation}\label{eq:extrarc}
\small
r_{\mathbb{C}}(A,D,E) = r_\mathbb{R}(A,D,E) = \frac{1}{\rho(E(-A)^{-1}DS)},
\end{equation}
where $\rho(.)$ denotes the spectral radius of the matrix.
\vspace{.2cm}

The proof of the theorem can be found in \cite{shafai1997explicit}. We use this result in our analysis of stability radius for positive Lur'e systems. 

\section{Stability Radius of Positive Lur'e Systems}
In this section, we formally define our problem and introduce our main result. We demonstrate and prove how we can utilize characteristics of Metzler matrices in linear systems for handling problems of nonlinear systems. Specifically, we study the stability radius of nonlinear Lur'e systems with structured perturbations under positivity constraints.
Consider the positive Lur'e system described by \eqref{eq:positiveluresystem},
where the nonlinearity $\Phi(.)$ satisfies the sector condition \eqref{eq:mimo sector bound} in $[\Sigma_1, \Sigma_2]$, with $\Sigma_1, \Sigma_2 \in \mathbb{R}^{m \times p}$ and $\Sigma_1 \leq \Sigma_2$.
We are interested in the robustness of this system to structured perturbations of the form:
\begin{equation}\label{eq:perturbed_system}
\small
\dot{x}(t) = \left( A + D \Delta E \right) x(t) + B \Phi(C x(t)),
\end{equation}
where $D \in \mathbb{R}_+^{n \times k}$, $E \in \mathbb{R}_+^{k \times n}$, and $\Delta \in \mathbb{R}^{k \times k}$ is an unknown perturbation matrix. Specifically, our objective is to determine the stability radius of system \eqref{eq:perturbed_system}. We denote the stability radius of the Lur'e system by
\begin{align}\label{eq:lurestabrad} \small
r_{\mathbb{R}}(\Phi) &= \inf\{\|\Delta\|: \left( A + D \Delta E \right) x(t)\\&+ B \Phi(C x(t)) \text{ is not globally asymptotically stable.}\}\nonumber
\end{align}
Before stating the first main result of the paper, we introduce a key lemma that is essential for its proof. The lemma characterizes the monotonicity of stability radius for Metzler and Hurwitz matrices.

\begin{lemma}[Monotonicity of Stability Radius]
\label{lem:StabRadiusMonotonicity}
Let $P, Q \in \mathbb{R}^{n \times n}$ be \emph{Metzler} and \emph{Hurwitz stable} matrices. Consider structured perturbations of the form:
\[
   P + D \,\Delta\,E
   \quad\text{and}\quad
   Q + D \,\Delta\,E,
\]
where the perturbation matrices $D \in \mathbb{R}_+^{n \times k}$, $E \in \mathbb{R}_+^{k \times n}$, and $\Delta \in \mathbb{R}^{k \times k}$. Given that $P \ge Q$, then
\begin{equation}\small
   r_{\mathbb{R}}(P)
   \;\le\;
   r_{\mathbb{R}}(Q),
\end{equation}
where $r_\mathbb{R}(.)$ is calculated by \eqref{eq:robustradii} with $\|\cdot\|$ being a chosen operator norm that is monotone on the cone of nonnegative matrices (e.g., $\|\cdot\|_1, \|\cdot\|_{\infty}$).
\end{lemma}

\begin{proof}
Since $P$ is Metzler and Hurwitz, $-P$ has nonpositive off-diagonals and eigenvalues with positive real parts. This implies that $-P$ is an \textit{M}-Matrix. The same applies to $-Q$.

Given $P \ge Q$, we have $-P \le -Q$. Therefore, Based on a classical property of nonsingular \textit{M}-matrices 
\[
 (-P)^{-1},(-Q)^{-1}\ge0 \quad \text{and}\quad  (-P)^{-1} \;\ge\; (-Q)^{-1}.
\]
Moreover, since $D$ and $E$ are nonnegative, left and right multiplication by these matrices preserves the order for nonnegative matrices:
\begin{equation*}\small
   E\,(-P)^{-1}D
   \;\ge\;
   E\,(-Q)^{-1}D.
\end{equation*}
Since we use a monotone operator norm $\|\cdot\|$, it follows that
\begin{equation*}\small
   \bigl\|\,E\,(-P)^{-1}D\bigr\|
   \;\ge\;
   \bigl\|\,E\,(-Q)^{-1}D\bigr\|.
\end{equation*}
Taking reciprocals and removing the negative sign inside the norm yield the following inequality for the stability radii:
\begin{equation*}\small
   r_{\mathbb{R}}(P)
   \;=\;
   \frac{1}{\bigl\|\,E\,(P)^{-1}D\bigr\|}
   \;\;\le\;\;
   \frac{1}{\bigl\|\,E\,(Q)^{-1}D\bigr\|}
   \;=\;
   r_{\mathbb{R}}(Q).
\end{equation*}
\end{proof}

\begin{remark}
We can also explain Lemma \ref{lem:StabRadiusMonotonicity} in different words. It is well known that the interval positive systems $[Q,P]$ are Hurwitz stable if and only if $P$ is Hurwitz stable. This has been established in the theorem $4$ of reference \cite{shafai1997explicit} for positive discrete time systems, which is equally applicable for continuous-time positive systems, as it is shown in \cite{shafai1990necessary}.
Now, without loss of generality, assume $D = E = I$. Since $P>Q$ and $P$ is asymptotically stable, it is closest to instability due to the continuity of stable eigenvalues\cite[Lemma 3.5]{bill2016stability}. Thus, the distance of $P$ from the set of singular matrices is given by $\|P^{-1}\|^{-1}$, which is the smallest singular value of $P$ with respect to $2-$norm. In other words,
$r_\mathbb{C}=r_\mathbb{R} = \sigma_{\min} =\|P^{-1}\|^{-1},$
where $\sigma$ denotes the singular value of matrix $P$.
\end{remark}

Now, with the aid of Lemma \ref{lem:StabRadiusMonotonicity}, we introduce the following theorem that provides a result on the stability radius of positive Lur'e systems subjected to structured perturbations.

\begin{theorem}\label{theorem:stability_radius}
Consider a Lur'e system described in \eqref{eq:perturbed_system}. Given that $A + B \Sigma_1 C$ is Metzler and $A + B \Sigma_2 C$ is Hurwitz,
 the stability radius of positive Lur'e system $r_{\mathbb{R}}(\Phi)$ with respect to the structured perturbation $D\Delta E$ is given by
\begin{equation}\label{eq:stability_radius}
\small
r_{\mathbb{R}}(\Phi) = \frac{1}{\left\| E (A + B \Sigma_2 C)^{-1} D \right\|}.
\end{equation}
\vspace{.01cm}
\end{theorem}

This value represents the norm of the largest allowable perturbation $\Delta$ such that the perturbed system remains globally exponentially stable.

\begin{proof}
Based on positive Aizerman conjecture and remark \ref{rem:stabrad}, the stability radius of the nonlinear Lur'e system can be measured by the stability radius of all the linear systems of the form $A+B\Sigma C, \Sigma\in[\Sigma_1,\Sigma_2]$ where $\Phi$ is sector bounded in the same interval. In other words
\begin{equation}\label{eq:eqmid1}\small
r_{\mathbb{R}}(\Phi) = \min_{\Sigma\in[\Sigma_1,\Sigma_2]} r_{\mathbb{R}}(A+B\Sigma C).
\end{equation}
Since $\forall\Sigma \in [\Sigma_1,\Sigma_2] \quad A + B \Sigma C$ is Hurwitz and Metzler and $B,C \geq 0$, we can utilize the result of Lemma \ref{lem:StabRadiusMonotonicity} and say
\begin{align}
\min_{\Sigma\in[\Sigma_1,\Sigma_2]} r_{\mathbb{R}}(A+&B\Sigma C) = r_{\mathbb{R}}(\max_{\Sigma\in[\Sigma_1,\Sigma_2]}(A+B\Sigma C))\nonumber\\
&=r_\mathbb{R}(A+B\Sigma_2 C). \label{eq:eqmid2}
\end{align}
Now, \eqref{eq:eqmid1} and \eqref{eq:eqmid2} result in:
\begin{equation*}\small
r_{\mathbb{R}}(\Phi) = \frac{1}{\left\| E (A + B \Sigma_2 C)^{-1} D \right\|}.
\end{equation*}
 
Since the nonlinear function $\Phi$ lies in the sector $[\Sigma_1, \Sigma_2]$, and $A + B \Sigma_2 C$ is Hurwitz, based on the implications of Remark \ref{rem:stabrad}, the same stability radius applies to the Lur'e system.
\end{proof}

It should be pointed out that the above theorem can also be stated in terms of other uncertainty structures, such as the one stated in Theorem \ref{the:Metstabrad} with the aid of equation \eqref{eq:extrarc}.

\section{Stability Radius of Lur'e Based Neural Network Feedback Loops}
In this section, we introduce a sector bound for fully connected FFNNs. This formulation enables the transformation of an NN-controlled system into the framework of a Lur'e system. By doing so, we can leverage established results from Lur'e system theory to assess the robust stability of systems controlled by NNs.
Consider a Lur'e system consisting of an LTI system \eqref{eq:generallti} and an FFNN in feedback loop denoted by \(\pi : \mathbb{R}^p \mapsto \mathbb{R}^m\). Suppose an FFNN with \(q\) layers described by the following equations:

\vspace{-.3cm}
\begin{subequations}\label{eq:NNcontroller}
\small
\begin{align}
    &\omega^{(0)}(t) = Cx(t), \\
    &\omega^{(i)}(t) = \phi^{(i)}( W^{(i)}\omega^{(i-1)}(t) + b^{(i)}), \quad i = 1, \dots, q, \label{eq:qthlayer} \\
    &u(t) = W^{(q+1)}\omega^{(q)}(t) + b^{(q+1)},
\end{align}
\end{subequations}

where \(\omega^{(i)}(t) \in \mathbb{R}^{l_i}\) represents the output of the \(i\)-th layer, with \(l_0 = p\) and \(l_{q+1} = m\), dictated by the dimensions of the LTI system’s input and output.

Each layer in the network is characterized by a weight matrix \(W^{(i)} \in \mathbb{R}^{l_i \times l_{i-1}}\), a bias vector \(b^{(i)} \in \mathbb{R}^{l_i}\), and an activation function \(\phi^{(i)}\), which is applied elementwise to its argument.

The set \((x_*, y_*, u_*)\) denotes the equilibrium state of the Lur'e system and satisfies the following equation:
\vspace{-.4cm}

{\small
\begin{equation*}
\mathbf{0}_n = A x_*(t) + B u_*(y_*), \quad y_*(t) = C x_*(t), \quad u_*(t) = \pi(Cx_*(t)).
\end{equation*}}

We now present a result regarding the sector bounds of general FFNNs as in \eqref{eq:NNcontroller}.
We assume the following properties:

\textit{\textbf{Property 1:}}
\begin{itemize}
    \item For simplicity, the activation functions for all neurons are assumed identical, and the biases are set to zero. Both of these conditions can be relaxed with minor modifications to the notation.
    \item The selected activation function lies within a sector \([a_1, a_2]\), where \(a_1 < a_2\)\footnote{Most famous NN activation functions are sector bounded in $[0,1]$ \cite{fazlyab2021introduction}.}. Define \(c = \max(|a_1|, |a_2|)\).
\end{itemize}

\begin{theorem}[\textbf{FFNN Sector Bound}\cite{10561541}]\label{the:NNsect}
Consider a fully connected FFNN \(\pi(\cdot)\) with \(q\) layers as defined in \eqref{eq:NNcontroller} satisfying \textit{\textbf{Property 1}}. The network accepts an input vector \(z \in \mathbb{R}^p_+\) and produces an output \(u(z) \in \mathbb{R}^m\). The given NN is sector bounded within \([\Gamma_1, \Gamma_2]\), where
\begin{equation}\label{eq:gammas}
\small
\Gamma_1 = -c^q \left( \prod_{i=1}^{q+1} |W^{(i)}|_{\text{abs}} \right), \quad \Gamma_2 = c^q \left( \prod_{i=1}^{q+1} |W^{(i)}|_{\text{abs}} \right),
\end{equation}
with \(|W^{(i)}|_{\text{abs}}\) being the elementwise absolute value of \(W^{(i)}\).
\end{theorem}

This result ensures that the output of the FFNN is constrained within a sector. The proof can be found in \cite{10561541}.

The proposed sector bounds enable us to utilize the results from Lur'e system analysis for robust stability analysis of uncertain NN-controlled systems. The following theorem establishes an explicit formula for the stability radius of linear systems with NNs in their feedback loop.
\begin{theorem}\label{corollary:NN_stability_radius} Consider system \eqref{eq:perturbed_system} in which the nonlinearity $\Phi$ is realized by a fully connected FFNN $\pi: \mathbb{R}^p \rightarrow \mathbb{R}^m$, satisfying the sector condition within $[\Gamma_1, \Gamma_2]$, as established in \eqref{eq:gammas}.
Given $A + B \Gamma_1 C$ is Metzler and $A + B \Gamma_2 C$ is Hurwitz, the stability radius $r_{\mathbb{f}}$ of the NN-controlled system with respect to the structured perturbation $D\Delta E$ is
\begin{equation}\label{eq:NN_stability_radius}
\small
r_{\mathbb{f}} = \frac{1}{\left\| E (A + B \Gamma_2 C)^{-1} D \right\|}.
\end{equation}
This value represents the maximal perturbation norm under which the NN-controlled system remains globally exponentially stable.
\end{theorem}

\begin{proof}
Given the structure of $\Gamma_2$ as described in Theorem \ref{the:NNsect} and the fact that $\Gamma_2 \geq \Gamma_1$, the proof follows the same steps as the proof of Theorem \ref{theorem:stability_radius} by substituting $\Sigma_2$ with $\Gamma_2$. Since the FFNN satisfies the sector condition within $[\Gamma_1, \Gamma_2]$ and $A + B \Gamma_2 C$ is Metzler and Hurwitz, we can apply the same analysis. Finally, we can use Lemma \ref{lem:StabRadiusMonotonicity} to say
\begin{equation}\label{eq:eqmidnn}
\small
     \min_{\Gamma\in[\Gamma_1,\Gamma_2]} r_{\mathbb{R}}(A+B\Gamma C) = \frac{1}{\left\| E (A + B \Gamma_2 C)^{-1} D \right\|}.
\end{equation}
\end{proof}

\begin{corollary}\label{cor:sectorbound}
If the exact upper sector bound $\Gamma_2$ for the NN $\pi$ is unknown or conservative, it can be refined using the exact stability radius $r_{\mathbb{f}}$ as
\begin{equation}\label{eq:refined_gamma}
\small
\|\Gamma_{2,\text{refined}}\| = \frac{1}{\left\| C (A + D \Delta E)^{-1} B \right\|},
\end{equation}
where $\Delta = r_{\mathbb{f}}$.

This analysis is very useful in the case of SISO systems, where $\|.\| = \absElem{.}$, limiting the choice of $\Gamma_{2,\text{refined}}$ to $\pm\|\Gamma_{2,\text{refined}}\|$.
In this methodology, one can find the destabilizing $\Delta = r_{\mathbb{f}}$, by incrementally increasing the perturbation to an NN-controlled system and recording the critical perturbation. This critical amount can be used in \eqref{eq:refined_gamma} to calculate the tightest sector bound for the NN.
This estimation provides a tighter upper bound $\Gamma_{2,\text{refined}}$, potentially improving the stability margin of the NN-controlled system.
\end{corollary}
\section{Examples}
In this section, we provide illustrative examples of the stability radius analysis using the results of the previous sections. We consider two cases: (1) a Lur'e system with structured uncertainty and (2) an instance of an NN-controlled system with structured uncertainty where the exact sector bound for the NN is unknown. We demonstrate that the sector bounds obtained using \eqref{eq:gammas} for the NN are conservative in this case. To address this, we refine these bounds and derive the tightest sector bounds, providing a more accurate analysis.

\subsection{Lur'e System with Structured Perturbations}
Consider a linear Metzler, open-loop unstable system characterized by the following matrices:
\begin{equation*}\small
A = \begin{bmatrix}
-5 & 5 & 1 \\
6  & -7 & 1\\
2 & 1 & -5
\end{bmatrix}, \quad 
B = \begin{bmatrix}
1 \\1\\1
\end{bmatrix}, \quad
C = \begin{bmatrix}
1 & 1 & 1
\end{bmatrix}.
\end{equation*}

A static nonlinear function $f(y) = -1.5y + 0.01y^3 + \sin(2y)$ tailored to be sector bounded in the interval \([\Sigma_1,\Sigma_2] = [-2,-0.48]\) was added to the feedback loop of the system. A schematic of this static nonlinear feedback is shown in Fig. \ref{fig:nonlin}. 
\begin{figure}
    \centering
    \includegraphics[width=.6\linewidth]{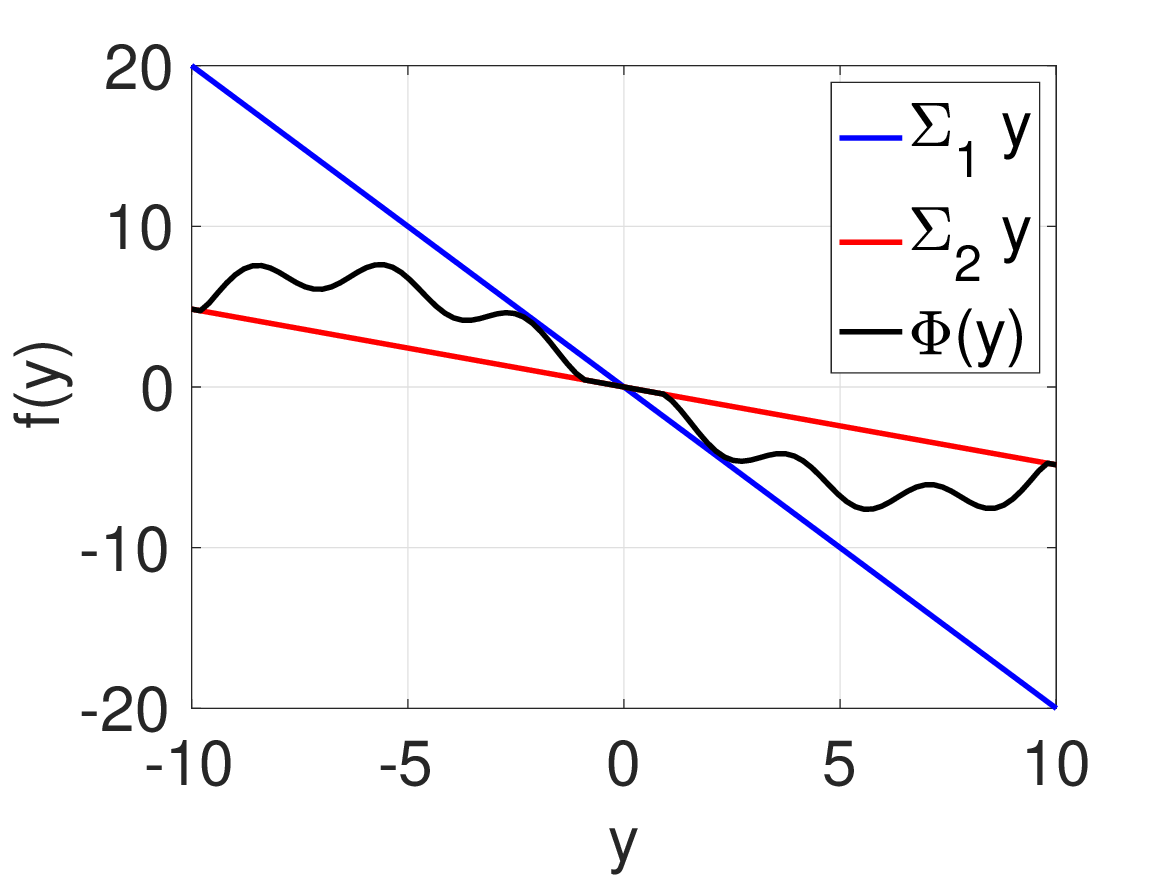}
    \caption{\small A schematic of the static nonlinearity tailored to lie within the sector bound $[\Sigma_1,\Sigma_2]$.\vspace{-.5cm}}
    \label{fig:nonlin}
\end{figure}
The closed-loop system satisfies the conditions of Theorem \ref{theorem:stability_radius} with this \(\Sigma_1\) and \(\Sigma_2\). For such a system, given \(D = \begin{bmatrix}
    1 & 0.5 & 1
\end{bmatrix}^T\), and \(E = \begin{bmatrix}
    0.5 & 1 & 1
\end{bmatrix}\), we expect the stability radius to be
\vspace{-.1cm}

{\small\[
r_{\mathbb{R}}(\Phi) = \frac{1}{\left\| E(A + B \Sigma_2 C)^{-1} D \right\|_2} = 0.26.
\]}

Fig. \ref{fig:Luresys} depicts the system's random trajectories as it transitions from stability to instability due to increasing structured perturbation. It is shown that the system remains exponentially stable until the threshold \(\Delta = 0.26\), beyond which it loses exponential stability.
\begin{figure}
    \centering
    \includegraphics[width=.9\linewidth]{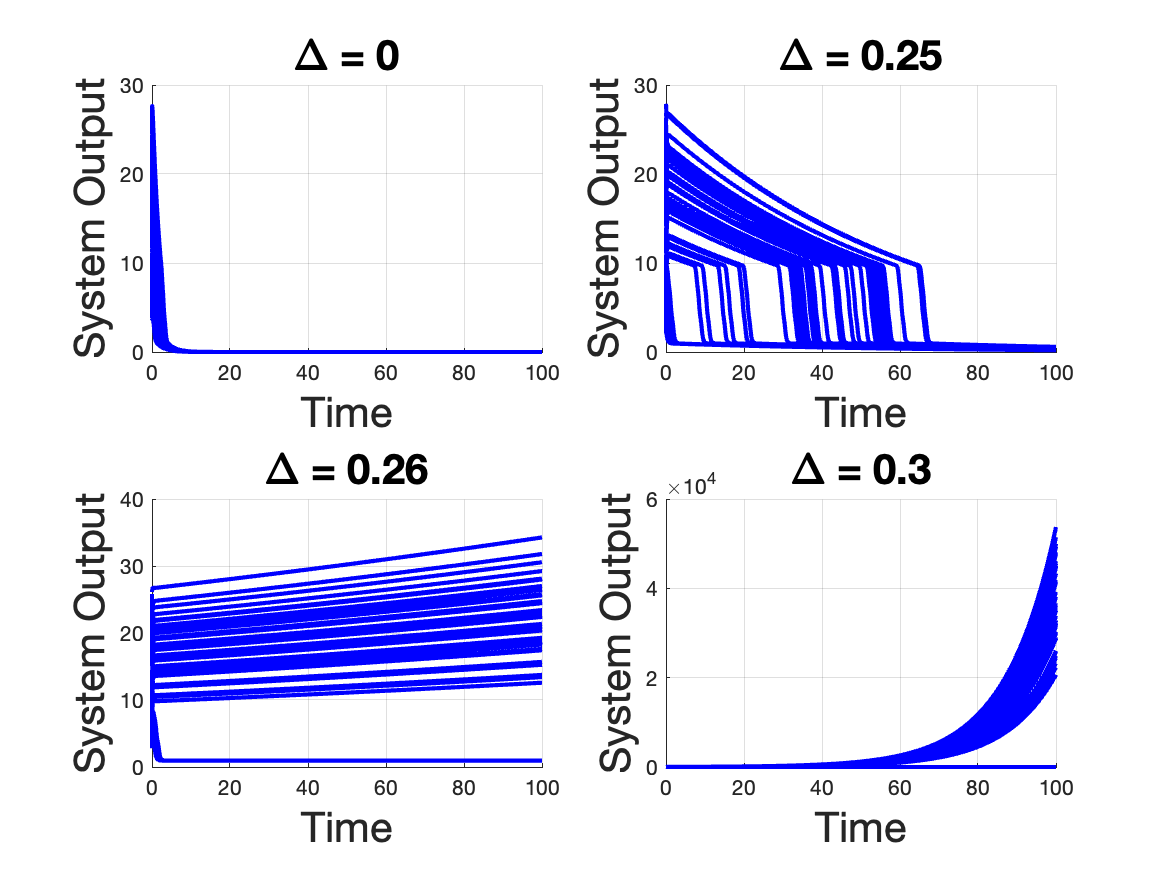}
    \caption{\small Random output trajectories of the system for various $\Delta$ values. The system preserves stability for $\Delta \leq 0.26$.}
    \label{fig:Luresys}
\end{figure}
\vspace{-.1cm}
\subsection{Neural Network-Controlled System}
Consider a linear system characterized by the matrices:
\begin{equation*}\small
A = \begin{bmatrix}
-5 & 3 & 1 \\
2 & -5  & 1 \\
3 & 1 & -4
\end{bmatrix}, \quad 
B = \begin{bmatrix}
0.5 \\
1 \\
0.4
\end{bmatrix}, \quad
C = \begin{bmatrix}
0.3 & 1 & 1
\end{bmatrix}.
\end{equation*}

The system is controlled by an FFNN trained on the data of a gain feedback controller. Using Theorem \ref{the:NNsect}, the NN sector bounds are calculated to be \([\Gamma_1,\Gamma_2]=[-0.91,0.91]\). Since $A+B\Gamma_1 C$ is Metzler and $A + B\Gamma_2 C$ is Hurwitz, the closed loop Lur'e system is positive and stable. Based on Theorem \ref{corollary:NN_stability_radius}, given \(D = \begin{bmatrix}
    1&0&0
\end{bmatrix}^T\), and \(E = \begin{bmatrix}
    1 & 0 & 0
\end{bmatrix}\), the stability radius of the system with the NN controller is
\(r_{\mathbb{f}} = 2.04\).
Fig. \ref{fig:NNinloop} shows random system trajectories as the perturbation increases and the system loses stability.
\begin{figure}
    \centering
    \includegraphics[width=.9\linewidth]{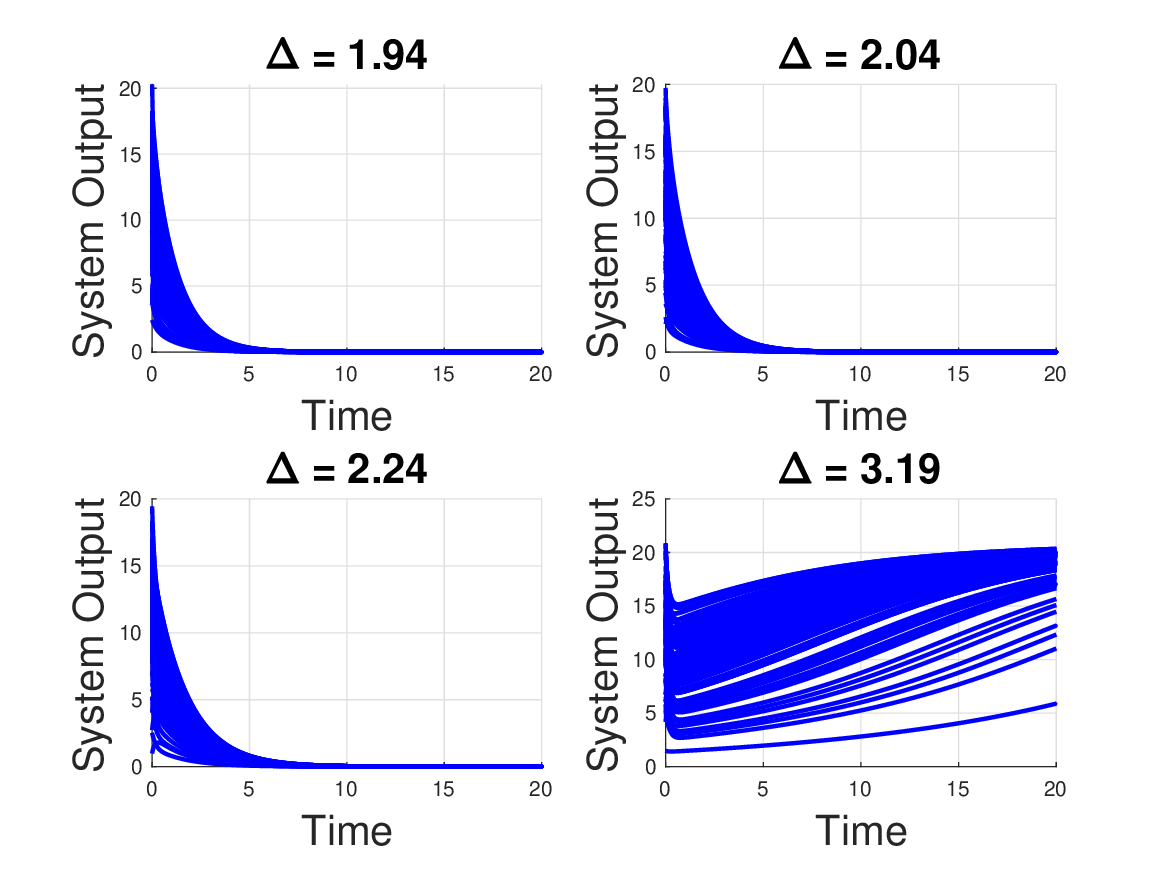}
    \caption{\small Random output trajectories of the system for various $\Delta$ values. The system is stable for $\Delta\leq2.04$ and transfers to instability with larger amounts of perturbations. This observation indicates that the NN upper sector bound $\Gamma_2$ might not be tight.}
    \label{fig:NNinloop}
\end{figure}

Note that the sector bounds calculated using Theorem \ref{the:NNsect} may be conservative in some cases, leading to a potentially conservative estimate for 
\(r_\mathbb{f} = 2.04\). 
This is particularly evident in the third diagram of Fig. \ref{fig:NNinloop}, where the system, even under a perturbation of \(\Delta = 2.24\) remains stable. This observation suggests that the upper limit of the NN sector bound, \(\Gamma_2\), may not be sufficiently tight.

To refine \(\Gamma_2\), we apply the procedure described in Corollary~\ref{cor:sectorbound}, incrementally increasing \(\Delta\) until the system becomes unstable. The corresponding critical value of \(\Delta\) is then substituted into equation~\eqref{eq:refined_gamma} to obtain a more accurate \(\Gamma_2\).

This process is illustrated in Figs. \ref{fig:unstabilizingD} and \ref{fig:uppergamma}. As shown in Fig. \ref{fig:unstabilizingD}, the system becomes unstable when the perturbation reaches \(\Delta = 3.15\). Fig. \ref{fig:uppergamma} shows the output of NN for random positive inputs and the recalculated upper bound \(\Gamma_{2,\text{refined}} z\). It is demonstrated that the NN output now falls outside the sector bound for the calculated $\Gamma_{2,\text{refined}}$ using the critical amount of \(\Delta = 3.15\) in equation \eqref{eq:refined_gamma}.

\begin{figure}
    \centering
    \includegraphics[width=.9\linewidth]{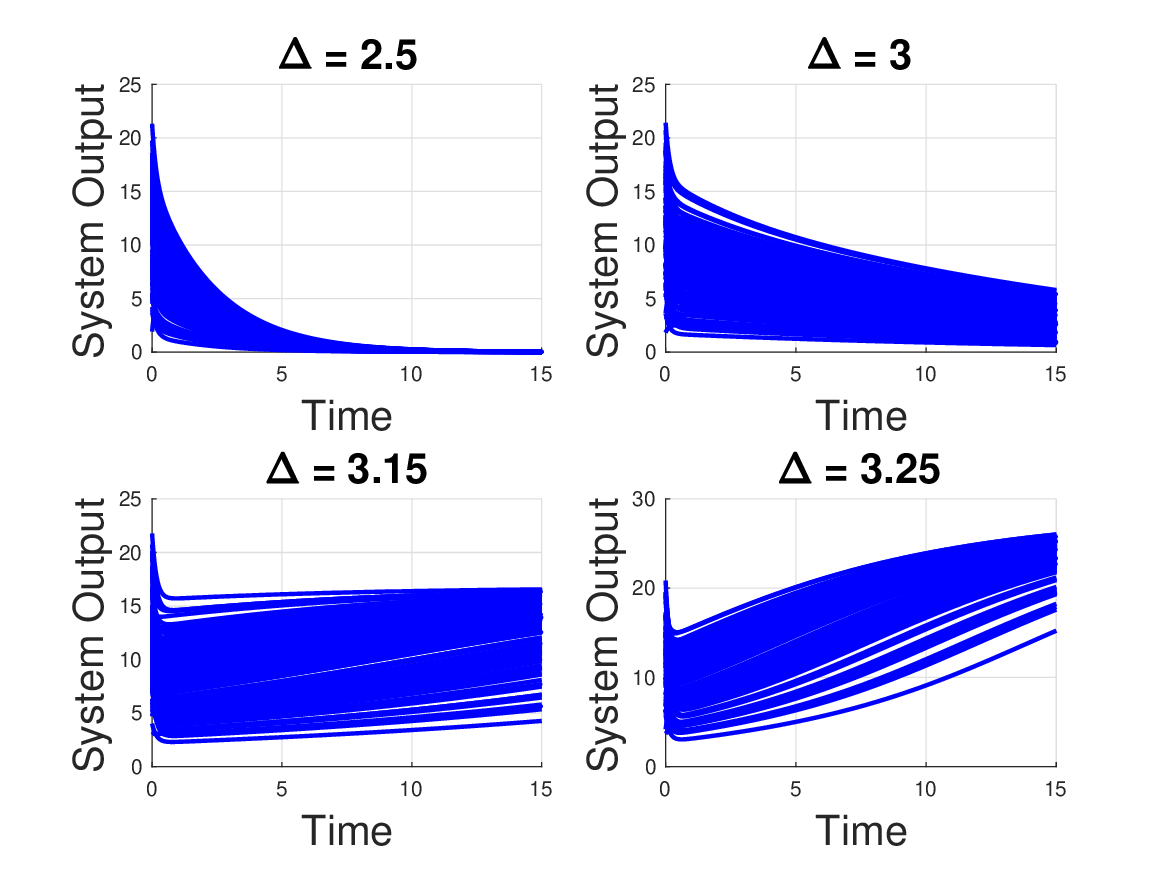}
    \caption{\small  Random output trajectories of the system for various $\Delta$ values. As illustrated, the system becomes unstable at \(\Delta = 3.15\).}
    \label{fig:unstabilizingD}
\end{figure}

\begin{figure}
    \centering
    \includegraphics[width=.9\linewidth]{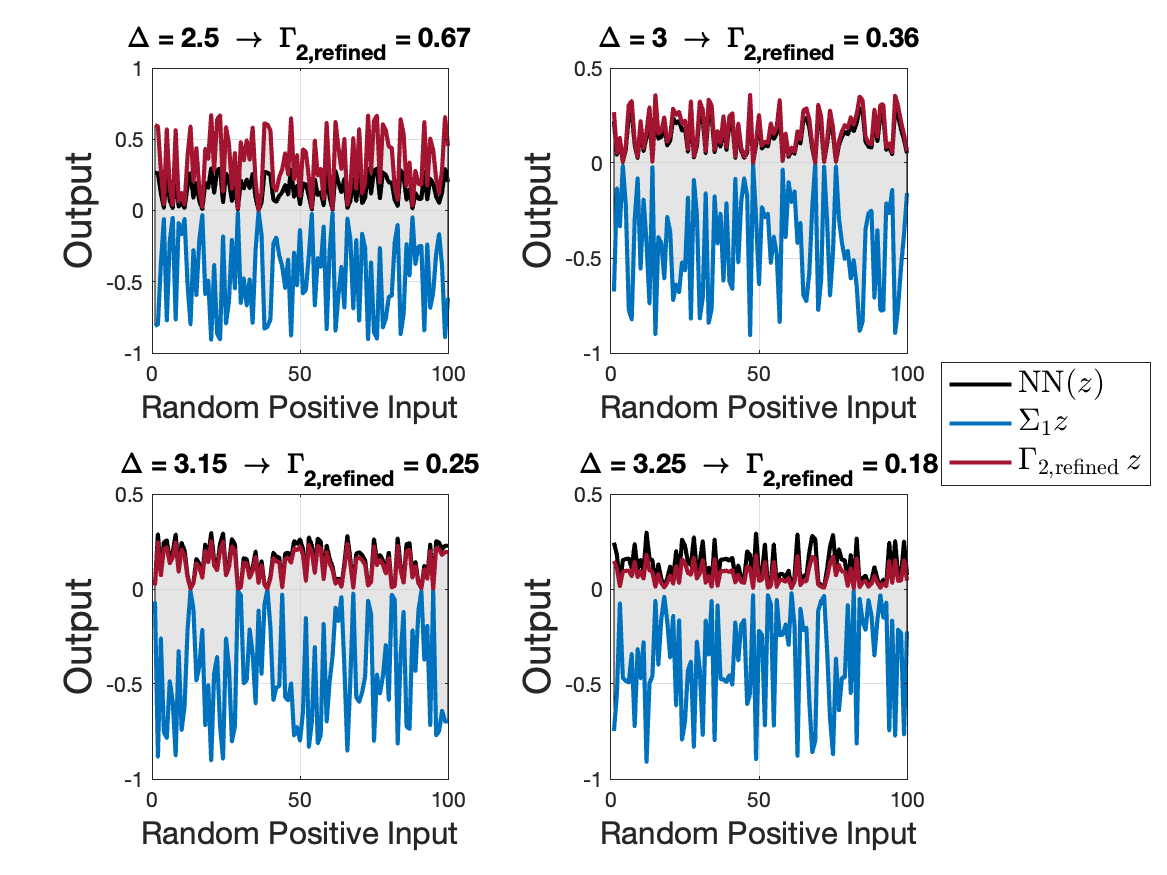}
    \caption{\small Plot of NN output vs. sector bounds for random positive inputs. Each figure shows the calculated value \(\Gamma_{2,\text{refined}}\) for the corresponding \(\Delta\) in Fig. \ref{fig:unstabilizingD}, using equation \eqref{eq:refined_gamma}. As illustrated here, the NN output falls out of sector bound for the amount of $\Gamma_{2,\text{refined}} = 0.25$ which corresponds to $\Delta = 3.15$.}
    \label{fig:uppergamma}
\end{figure}

\section{Discussion and Conclusion}
By leveraging positivity constraints, we utilize the Aizerman conjecture and the stability radius of Metzler matrices to evaluate the robustness of nonlinear Lur'e systems. The explicit expressions for the stability radius derived in Theorem \ref{theorem:stability_radius} and Theorem \ref{corollary:NN_stability_radius} define the robustness of positive Lur'e systems and NN-controlled systems against structured uncertainties. Corollary \ref{cor:sectorbound} also proposes a method for finding the tightest sector bounds for NNs. These results highlight several key insights:

\textbf{Dependence on the Upper Sector Bound:} The stability radius of nonlinear Lur'e systems is determined by the worst-case (upper) sector bound ($\Sigma_2$ or $\Gamma_2$). This underscores the importance of accurately characterizing the sector in which the nonlinearity or NN operates.

\textbf{Limitation:} Positivity is not imposed on the linear system; the matrix $A$ may be non-Metzler. Thus, the analysis covers a broad class of systems. It is only required that the closed-loop nonlinear system be positive through $A+B\Sigma_1C.$

\textbf{NN bound refinement:} Unlike other bounding methods such as CROWN and IBP, which are sound but incomplete, our sound and complete method yields an exact upper sector bound for the NN in the loop. It requires no knowledge of the NN parameters (weights, biases), treating the NN as a black box and shifting the analysis to the linear system—greatly reducing computational effort.

\textbf{Monotone Norms:}
Based on Lemma \ref{lem:StabRadiusMonotonicity}, perturbations are quantified using norms that are monotonic on the cone of nonnegative matrices. Notable examples of these norms include the 1-norm and the infinity norm. These norms are widely used to measure noise and uncertainties within the system, making them well-suited for our analysis.

\textbf{Role of System Matrices:} The matrices $A$, $B$, and $C$ directly influence the stability radius through the term $(A + B\Sigma_2 C)^{-1}$ or $(A + B\Gamma_2 C)^{-1}$. Well-designed system matrices can enhance robustness by reducing the norm $\left\| E (A + B\Sigma_2 C)^{-1} D \right\|_2$.

\textbf{Applicability to NNs:} By modeling the NN controller as a sector-bounded nonlinearity, we bridge the gap between traditional control theory and modern NN-based controllers. The derived stability radius provides a quantitative measure for the robustness of NN-controlled systems.

Future work may include the analysis of alternative perturbation structures, the integration of additional results from positive linear systems theory into the Lur'e framework, and incorporation of other NN architectures in the feedback loop.
\vspace{-.1cm}
\begin{spacing}{.91}
\bibliography{Reference1}

\begin{thebibliography}{10}

\bibitem{szegedy2013intriguing}
C.~Szegedy, W.~Zaremba, I.~Sutskever, J.~Bruna, D.~Erhan, I.~Goodfellow, and R.~Fergus, ``Intriguing properties of neural networks,'' {\em arXiv preprint arXiv:1312.6199}, 2013.

\bibitem{dawson2023safe}
C.~Dawson, S.~Gao, and C.~Fan, ``Safe control with learned certificates: A survey of neural {L}yapunov, barrier, and contraction methods for robotics and control,'' {\em IEEE Transactions on Robotics}, 2023.

\bibitem{tsukamoto2021contraction}
H.~Tsukamoto, S.-J. Chung, and J.-J.~E. Slotine, ``Contraction theory for nonlinear stability analysis and learning-based control: A tutorial overview,'' {\em Annual Reviews in Control}, vol.~52, pp.~135--169, 2021.

\bibitem{TESI1991147}
A.~Tesi and A.~Vicino, ``Robust absolute stability of lur'e control systems in parameter space,'' {\em Automatica}, vol.~27, no.~1, pp.~147--151, 1991.

\bibitem{wang2008reliable}
H.~Wang, A.~Xue, R.~Lu, and J.~Wang, ``Reliable robust h$\infty$ tracking control for lur’e singular systems with parameter uncertainties,'' in {\em 2008 American Control Conference}, pp.~4312--4317, IEEE, 2008.

\bibitem{hao2010absolute}
F.~Hao, ``Absolute stability of uncertain discrete lur’e systems and maximum admissible perturbed bounds,'' {\em Journal of the Franklin Institute}, vol.~347, no.~8, pp.~1511--1525, 2010.

\bibitem{tan1999absolute}
N.~Tan and D.~P. Atherton, ``Absolute stability problem of systems with parametric uncertainties,'' in {\em 1999 European Control Conference (ECC)}, pp.~2829--2834, IEEE, 1999.

\bibitem{yin2021stability}
H.~Yin, P.~Seiler, and M.~Arcak, ``Stability analysis using quadratic constraints for systems with neural network controllers,'' {\em IEEE Transactions on Automatic Control}, vol.~67, no.~4, pp.~1980--1987, 2021.

\bibitem{fuh2009robust}
C.-C. Fuh and P.-C. Tung, ``Robust stability bounds for lur'e system with parametric uncertainty,'' {\em Journal of Marine Science and Technology}, vol.~7, no.~2, p.~1, 2009.

\bibitem{dun2008robust}
A.~Dun, Z.~Geng, and L.~Huang, ``Robust dichotomy of the lur'e system with structured uncertainties,'' {\em International Journal of Control}, vol.~81, no.~5, pp.~778--787, 2008.

\bibitem{farina2000positive}
L.~Farina and S.~Rinaldi, {\em Positive Linear Systems: Theory and Applications}.
\newblock Pure and Applied Mathematics: A Wiley Series of Texts, Monographs and Tracts, Wiley, 2000.

\bibitem{10708136}
B.~Shafai, F.~Zarei, and A.~Moradmand, ``Stabilization of input derivative positive systems and its utilization in positive singular systems,'' in {\em 2024 10th International Conference on Control, Decision and Information Technologies (CoDIT)}, pp.~615--620, 2024.

\bibitem{shafai2024positive}
B.~Shafai and F.~Zarei, ``Positive stabilization and observer design for positive singular systems,'' in {\em Proceedings of the 2024 63rd IEEE Conference on Decision and Control (CDC), Milan, Italy}, pp.~16--19, 2024.

\bibitem{fazlyab2021introduction}
M.~Fazlyab, M.~Morari, and G.~J. Pappas, ``An introduction to neural network analysis via semidefinite programming,'' in {\em 2021 60th IEEE Conference on Decision and Control (CDC)}, pp.~6341--6350, IEEE, 2021.

\bibitem{yin2021imitation}
H.~Yin, P.~Seiler, M.~Jin, and M.~Arcak, ``Imitation learning with stability and safety guarantees,'' {\em IEEE Control Systems Letters}, vol.~6, pp.~409--414, 2021.

\bibitem{10561541}
H.~Montazeri~Hedesh and M.~Siami, ``Ensuring both positivity and stability using sector-bounded nonlinearity for systems with neural network controllers,'' {\em IEEE Control Systems Letters}, vol.~8, pp.~1685--1690, 2024.

\bibitem{hedesh2025local}
H.~M. Hedesh, M.~K. Wafi, and M.~Siami, ``Local stability and region of attraction analysis for neural network feedback systems under positivity constraints,'' {\em arXiv preprint arXiv:2505.22889}, 2025.

\bibitem{Wafi}
M.~K. Wafi, M.~Siami, and M.~Sznaier, ``Investigating the effectiveness of reinforcement learning in closed-loop systems with time delays,'' in {\em 2024 American Control Conference (ACC)}, pp.~4149--4154, July 2024.

\bibitem{mashhadireza2024iterative}
A.~Mashhadireza and A.~Sadighi, ``Iterative learning control for friction compensation of a lorentz actuator for periodic references,'' in {\em 2024 12th RSI International Conference on Robotics and Mechatronics (ICRoM)}, pp.~697--702, IEEE, 2024.

\bibitem{10155284}
C.~R. Richardson, M.~C. Turner, and S.~R. Gunn, ``Strengthened circle and popov criteria for the stability analysis of feedback systems with relu neural networks,'' {\em IEEE Control Systems Letters}, vol.~7, pp.~2635--2640, 2023.

\bibitem{kotha2023provably}
S.~Kotha, C.~Brix, J.~Z. Kolter, K.~Dvijotham, and H.~Zhang, ``Provably bounding neural network preimages,'' in {\em Advances in Neural Information Processing Systems} (A.~Oh, T.~Neumann, A.~Globerson, K.~Saenko, M.~Hardt, and S.~Levine, eds.), vol.~36, pp.~80270--80290, Curran Associates, Inc., 2023.

\bibitem{zhang22babattack}
H.~Zhang, S.~Wang, K.~Xu, Y.~Wang, S.~Jana, C.-J. Hsieh, and Z.~Kolter, ``A branch and bound framework for stronger adversarial attacks of {R}e{LU} networks,'' in {\em Proceedings of the 39th International Conference on Machine Learning}, vol.~162, pp.~26591--26604, 2022.

\bibitem{singh2019abstract}
G.~Singh, T.~Gehr, M.~P{\"u}schel, and M.~Vechev, ``An abstract domain for certifying neural networks,'' {\em Proceedings of the ACM on Programming Languages}, vol.~3, no.~POPL, pp.~1--30, 2019.

\bibitem{zhang2019recurjac}
H.~Zhang, P.~Zhang, and C.-J. Hsieh, ``Recurjac: An efficient recursive algorithm for bounding jacobian matrix of neural networks and its applications,'' in {\em Proceedings of the AAAI Conference on Artificial Intelligence}, vol.~33, pp.~5757--5764, 2019.

\bibitem{sontag2013mathematical}
E.~D. Sontag, {\em Mathematical control theory: deterministic finite dimensional systems}, vol.~6.
\newblock Springer Science \& Business Media, 2013.

\bibitem{drummond2022aizerman}
R.~Drummond, C.~Guiver, and M.~C. Turner, ``Aizerman conjectures for a class of multivariate positive systems,'' {\em IEEE Transactions on Automatic Control}, 2022.

\bibitem{shafai1997explicit}
B.~Shafai, J.~Chen, and M.~Kothandaraman, ``Explicit formulas for stability radii of nonnegative and metzlerian matrices,'' {\em IEEE transactions on automatic control}, vol.~42, no.~2, pp.~265--270, 1997.

\bibitem{shafai1990necessary}
B.~Shafai, K.~Perey, D.~Cowley, and Y.~Chehab, ``A necessary and sufficient condition for the stability of nonnegative interval discrete systems,'' in {\em 1990 American Control Conference}, pp.~894--899, IEEE, 1990.

\bibitem{bill2016stability}
A.~Bill, C.~Guiver, H.~Logemann, and S.~Townley, ``Stability of nonnegative lur'e systems,'' {\em SIAM Journal on Control and Optimization}, vol.~54, no.~3, pp.~1176--1211, 2016.

\end{thebibliography}
\end{spacing}
\end{document}